\documentclass[aps,prl,reprint,superscriptaddress,nofootinbib]{revtex4-2}

\usepackage{amsmath,amsthm,amssymb,mathtools,bm}
\usepackage{color}

\usepackage{tikz}
\usepackage{hyperref}
\usetikzlibrary{decorations.pathreplacing}

\newcommand{\cF}{\mathcal F}
\newcommand{\cD}{\mathcal D}
\newcommand{\cH}{\mathcal H}
\newcommand{\cC}{\mathcal C}
\newcommand{\Tr}{\operatorname{Tr}}
\newcommand{\Fix}{\operatorname{Fix}}

\newcommand{\FSEP}{\mathcal F_{\mathrm{SEP}}}

\newcommand{\ket}[1]{|#1\rangle}
\newcommand{\bra}[1]{\langle #1|}

\theoremstyle{definition}
\newtheorem{assumption}{Assumption}
\newtheorem{capability}{Capability}
\newtheorem{proposition}{Proposition}
\newtheorem{lemma}{Lemma}
\newtheorem{theorem}{Theorem}
\newtheorem{corollary}{Corollary}
\newtheorem{remark}{Remark}

\begin{document}

\title{Inverse Problem of Alchemical Resource Theory: Replication and Universal Simulation Single Out Imaginarity and Parity Asymmetry}


\author{Yasuaki Nakayama}
\email{yasuaki.nakayama@ntt.com}
\affiliation{NTT Communication Science Laboratories, NTT, Inc., 3-1 Morinosato Wakamiya, Atsugi, Kanagawa 243-0198, Japan}

\author{Hayato Arai}
\email{h.arai6626@gmail.com, haya.arai@ntt.com}
\affiliation{NTT Communication Science Laboratories, NTT, Inc., 3-1 Morinosato Wakamiya, Atsugi, Kanagawa 243-0198, Japan}

\begin{abstract}
Quantum resource theories usually begin with a prescribed free structure and ask what tasks become possible when a resource is supplied. We study the inverse problem of alchemical resource theories, which we define as resource theories admitting a resource that can both replicate itself exactly and universally simulate quantum instruments. Under natural consistency assumptions and branchwise complete freeness, we show that for multi-qubit systems only two nontrivial theories survive, up to a common local unitary change of basis: parity asymmetry and imaginarity. We further show that exact universality exhibits an all-or-nothing trade-off: any nonmaximal resource state can exactly simulate only free unitaries. These results demonstrate that prescribed operational capabilities can strongly constrain the underlying resource structure and reveal a fundamental trade-off between computational capability and the precision required for physical implementation.
\end{abstract}

\maketitle


Quantum resource theories provide a general framework for resources in quantum information processing~\cite{ChitambarGour2019,BrandaoGour2015,Gour2017,Regula2018}, encompassing prominent examples such as entanglement~\cite{HorodeckiEtAl2009}, coherence~\cite{BaumgratzEtAl2014}, asymmetry~\cite{GourSpekkens2008}, athermality~\cite{BrandaoEtAl2013}, and magic~\cite{VeitchEtAl2014}.
Conventionally, one first specifies the free states and operations and then asks what can be achieved by supplying a resource.

However, from the viewpoint of designing quantum information-processing devices, the opposite direction is equally natural: the desired functionality may be known before the resource structure that realizes it.
Related inverse viewpoints have appeared in several forms, including reconstruction problems concerning different characterizations of resource theories~\cite{ScandiSurace2021}, deriving resource structures or operation classes from operational or inferential principles~\cite{NagasawaEtAl2025,LieEtAl2026}, and constraining general resource theories through prescribed operational requirements~\cite{SonEtAl2026}.
In particular, identifying further settings in which resource replication can arise has been explicitly left as a direction for future work~\cite{KuroiwaYamasaki2020}.
Here we pursue this viewpoint at the level of the resource theory itself, treating the resource structure as the unknown to be determined from prescribed operational capabilities (Fig.~\ref{fig:inverse-problem}).

\begin{figure}[t]
\centering
\vspace*{0.12cm}
\makebox[\columnwidth][c]{
\hspace*{0cm}
\begin{tikzpicture}[x=0.4cm,y=0.4cm,line width=0.55pt]
    \draw[rounded corners=6pt] (0,0) rectangle (7,5);
    \draw[rounded corners=6pt] (14,0) rectangle (21,5);
    \draw[->](8.7,2.4)--(12.3,2.4);
    \draw[line width=1pt,->](17.5,-0.4) .. controls (13.5,-1.9) and (7.5,-1.9) .. (3.5,-0.4);
    \draw (3.5,4.3) node {\textbf{Structure of}};
    \draw (3.5,3.5) node {\textbf{resource theory}};
    \draw (3.5,2.1) node {\footnotesize{free states}};
    \draw (3.5,1.1) node {\footnotesize{free operations}};
    \draw (10.5,3.2) node {\scriptsize{Standard direction}};
    \draw (17.5,4.3) node {\textbf{Operational}};
    \draw (17.5,3.5) node {\textbf{capabilities}};
    \draw (17.5,2.2) node {\footnotesize{resource replication}};
    \draw (17.5,1.3) node {\footnotesize{universal instrument}};
    \draw (17.5,0.62) node {\footnotesize{simulation}};
    \draw (10.5,-2.2) node {\textbf{Inverse problem}};
\end{tikzpicture}
}
\vspace{-0.35cm}
\caption{Conventional and inverse approaches to quantum resource theories. Conventional analysis starts from a prescribed resource structure and derives its operational capabilities, whereas the inverse problem asks which resource structures are compatible with prescribed operational capabilities.}
\label{fig:inverse-problem}
\end{figure}
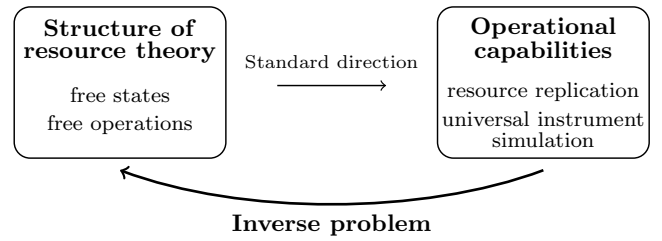

In this inverse setting, the classes of free states and free operations are themselves to be determined from the prescribed capabilities, which constrain the admissible resource structures.
This reverses the usual direction of resource-theoretic analysis, where the free structure is specified first and its operational consequences are studied afterward.
The problem is subtler for composite systems: local notions of freeness do not generally determine which joint states are free, as seen, for example, in resource theories of asymmetry and quantum reference frames~\cite{BartlettRudolphSpekkens2007,GourSpekkens2008,MiyazakiAkibue2024,WuEtAl2024Distributed} and in recent axiomatic studies of composite resource theories~\cite{GanardiEtAl2026}.
The inverse problem must therefore determine a consistent resource structure across both individual and composite systems.

For a resource to serve as a useful ingredient for scalable quantum information processing, it should be both repeatedly available and usable for general information-processing tasks.
We therefore impose two strong capabilities on a resource state \(g\): exact self-replication,
\(g\mapsto g^{\otimes 2}\), providing scalable resource supply, and exact simulation of arbitrary instruments on finite composite systems. As noted below, together with replication and CFK compositionality, it is enough to impose the simulation requirement only on instruments whose input and output each contain at most two single systems.
We call a resource state possessing both capabilities \emph{alchemical}, and a resource theory admitting such a resource an \emph{alchemical resource theory}.
Their combination represents a maximal exact operational capability and, in particular, enables universal quantum computation~\cite{BravyiKitaev2005}.
Closely related capabilities are known in both parity asymmetry and imaginarity~\cite{GourSpekkens2008,MarvianSpekkens2012,HickeyGour2018,KuroiwaYamasaki2020,WuEtAl2021,WuEtAl2021PRA,Takeuchi2024,NakayamaTakeuchiAkibue2026}.
However, whether these idealized resource structures can be naturally realized and exploited on practical physical platforms remains nontrivial.
We therefore ask whether other resource structures can provide the same capabilities while being more amenable to physical realization; any such solution would be highly valuable for quantum information processing.

Our main result shows that, among convex resource theories, alchemical resource theories are highly restricted.
For a non-trivial local free-state space, replication and universal simulation single out only two qubit resource structures, up to unitary equivalence: parity asymmetry and imaginarity.
While restrictive from the original design perspective, this classification reveals a more fundamental structure: without assuming any symmetry principle, the operational capabilities alone enforce a Wigner-type unitary--antiunitary dichotomy in the resource structure~\cite{Wigner1931,Bargmann1964}.
The derivation further uncovers another theoretically important feature: a hidden nonlocality in the composite free structure.
Self-replication already requires free composite states beyond mixtures of locally free product states, and the full alchemical requirements strengthen this further, forcing entangled states themselves to be free.

The classification also yields two consequences for physical implementation.
First, resource-nongenerating (RNG) operations are widely used as the largest class preserving free states, although they can be too permissive as physically free operations~\cite{ChitambarGour2019,ChitambarGour2016,LiuHuLloyd2017,AlbarelliEtAl2018,SaxenaEtAl2020,TajimaTakagi2025}.
Our setting provides a new example: RNG operations allow artificial self-replication by hiding branchwise resource generation through averaging.
Second, both admissible resource structures exhibit an all-or-nothing behavior in exact simulation: a maximally resourceful state can exactly simulate arbitrary qubit unitaries, whereas any nonmaximal resource can exactly simulate only free unitaries.
Hence maximal exact capability requires exact maximality of the resource state.
Together, these results reveal a new trade-off between operational capability and physical implementability, with strong exact capabilities requiring stringent conditions on both free implementations and resource-state preparation.

We first discuss how to formulate the inverse problem in an operationally meaningful way and introduce minimal natural assumptions.
The proof shows that replication and universal simulation force every local free state to be unbiased with respect to the resource state.
On the composite side, exact replication cannot be supported by mixtures of locally free product states alone, forcing genuinely global free correlations.
For qubits, the resulting local free-state geometry admits only two nontrivial possibilities, which propagate to all finite composite systems and yield the unitary and antiunitary involution theories.
The remainder of the main text develops its consequences for physical implementation.
The End Matter provides additional constructions and the proof of the replication-only correlation result, while the Supplemental Material provides the complete proof of the finite-qubit classification.

\paragraph{Inverse problem setting.---}
First of all, we specify the general setting of resource theories to be explored in our inverse problem. Now, we denote $\cD(\cH)$ as the set of all density matrices on a Hilbert space $\cH$.
For a single system, let \(\cF_1\subseteq\cD(\cH_1)\) be a nonempty closed convex set of free states,
and define its unnormalized free cone
\begin{equation}
 \cC_1:=\{t\sigma:t\geq0,\ \sigma\in\cF_1\}.
 \label{eq:cone}
\end{equation}
Here, we assume that all single systems have the same structure. When subsystem labels are explicit, we use the corresponding notation \(\cF_A,\cF_B,\cC_A,\cC_B,\ldots\).
For a composite system $AB$, we define $\cF_{AB}$ and $\cC_{AB}$ on the Hilbert space $\cH_A\otimes\cH_B$ in the same manner.
We further impose the following consistency conditions.

(C1) Free states are closed under independent composition:
\begin{equation}
    \sigma_A\in\cF_A,\ \tau_B\in\cF_B
    \Longrightarrow
    \sigma_A\otimes\tau_B\in\cF_{AB}.
    \label{eq:product-closure}
\end{equation}
This requires that independently combining free systems does not generate a resource.

(C2) Free states are closed under taking marginals:
\begin{equation}
    \rho_{AB}\in\cF_{AB}
    \Longrightarrow
    \Tr_B[\rho_{AB}]\in\cF_A,\ 
    \Tr_A[\rho_{AB}]\in\cF_B.
    \label{eq:marginal-consistency}
\end{equation}
Thus, simply discarding part of a free system cannot generate a resource.

(C3) Freeness is invariant under permutations of subsystems. For a permutation $\pi$ of the subsystems, with $U_\pi$ the corresponding permutation unitary, we assume
\begin{equation}
U_\pi\cF_{A_1\cdots A_n}U_\pi^\dagger
=
\cF_{A_{\pi(1)}\cdots A_{\pi(n)}}.
\label{eq}
\end{equation}
This ensures that freeness does not depend on the labeling or ordering of subsystems.

We next define the free operations considered in this work.
We do not adopt resource-nongenerating (RNG) operations as free operations, since they can admit operationally artificial realizations of the prescribed capabilities.
In particular, as shown explicitly in the End Matter, even a simple measure-and-prepare RNG channel can replicate a resource state.
This provides a new operational argument, complementary to previous observations that maximal resource-nongenerating operations can be too broad to represent physically meaningful free operations~\cite{ChitambarGour2019,ChitambarGour2016,LiuHuLloyd2017,AlbarelliEtAl2018,SaxenaEtAl2020,TajimaTakagi2025}.

Instead, we require freeness at the level of individual Kraus branches under arbitrary extensions. A Kraus operator $K:\cH_A\to\cH_B$ is called a \emph{completely free Kraus operator} if, for every reference system $R$,
\begin{equation}
(K\otimes I_R)\cC_{AR}(K^\dagger\otimes I_R)
\subseteq \cC_{BR},
\label{eq:cfk}
\end{equation}
where $I_R$ is the identity operator on $\cH_R$.
A completely positive map $\Phi:A\to B$ is called a \emph{Completely Free Kraus (CFK) operation} if
$\Phi(X)=\widetilde\Phi(X\otimes\eta_E)$ for some finite ancillary system $E$, free state $\eta_E\in\cF_E$, and completely positive map $\widetilde\Phi:AE\to B$ admitting a Kraus representation consisting entirely of completely free Kraus operators.
An instrument $\{\Phi_j:A\to B\}_j$ is called CFK if its outcome maps admit such realizations with a common finite free ancilla $\eta_E$, and $\sum_j\Phi_j$ is trace preserving. Throughout this work, we take the class of CFK operations defined above as the free operations.

Operationally, the CFK condition excludes resource generation even when post-selection over instrument outcomes is allowed and arbitrary free ancillary systems are available, including those correlated with untouched reference systems.
The CFK class is closed under sequential and tensor-product composition, outcome-conditioned composition, and arbitrary scalar multiplication of Kraus operators; proofs are given in the Supplemental Material.

For the reader's convenience, we summarize our setting.
\begin{assumption}[Inverse problem setting]
\label{assumptions}
Throughout this work, we consider resource theories satisfying the following conditions:
\begin{enumerate}
    \item[(A1)] The free-state set of each system is nonempty, closed, and convex, and all single systems have the same structure.

    \item[(A2)] The free-state structure satisfies the consistency conditions (C1)--(C3).

    \item[(A3)] The class of free operations is given by the CFK operations.
\end{enumerate}
\end{assumption}
These assumptions specify the class of resource theories over which we formulate our inverse problem. 
Next, we introduce the operational capabilities that define the problem.

Let $g=\ket{\gamma}\!\bra{\gamma}\in\cD(\cH_1)\setminus\cF_1$ be a pure resource state of a single system $P$. We impose the following two operational capabilities on $g$.

\begin{capability}[Resource replication]
\label{cap:replication}
There exists a CFK channel $\mathcal R$ such that
\begin{equation}
    \mathcal R(g)=g^{\otimes 2}.
    \label{eq:replication}
\end{equation}
\end{capability}

This capability provides a scalable supply of the resource.

\begin{capability}[Universal instrument simulation]
\label{cap:universal-simulation}
For every instrument $\{\mathcal N_j:A\to B\}_j$ between arbitrary finite composite systems \(A\) and \(B\), there exist a finite integer $k$ and a CFK instrument
$\{\Phi_j:A\otimes P^{\otimes k}\to B\}_j$ such that
\begin{equation}
    \Phi_j(\rho\otimes g^{\otimes k})
    =
    \mathcal N_j(\rho)
    \label{eq:universal-simulation}
\end{equation}
for every input state $\rho$ and every outcome $j$.
\end{capability}


We call a resource state \(g\) satisfying both capabilities \emph{alchemical}, and a resource theory admitting such a state an \emph{alchemical resource theory}.
An alchemical resource thus has maximal exact operational capability.

\begin{remark}
\label{rem}
Although Capability~\ref{cap:universal-simulation} is stated for arbitrary finite composite systems, in the presence of replication it is enough to require exact simulation only for instruments whose input and output each contain at most two single systems. Indeed, any finite-system instrument admits an exact finite Stinespring realization using state preparation, one- and two-system unitaries, single-system measurements and discarding, and classical coarse-graining. CFK operations are closed under the required tensor and sequential compositions, while replication supplies the finitely many copies of \(g\) used by the elementary simulators. Thus this two-system condition already implies Capability~\ref{cap:universal-simulation}.
\end{remark}

The two capabilities are illustrated in Fig.~\ref{fig:capabilities}, with the unitary case shown for simulation.

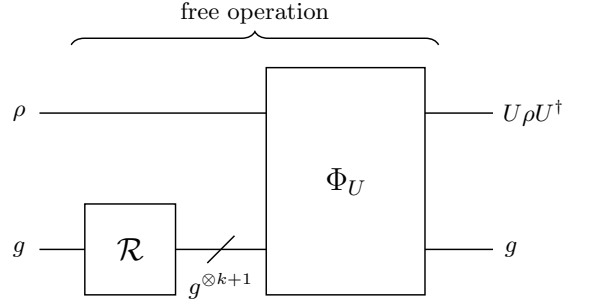
\begin{figure}[t]
\centering
\makebox[\columnwidth][c]{
\hspace*{0.33cm}
\begin{tikzpicture}[x=0.6cm,y=0.6cm,line width=0.55pt]
    \draw (1,0) rectangle (3,2);
    \draw (5,0) rectangle (8.5,5);
    \draw(0,1)--(1,1);
    \draw(3,1)--(5,1);
    \draw(3.7,0.7)--(4.3,1.3);
    \draw(8.5,1)--(10,1);
    \draw(0,4)--(5,4);
    \draw(8.5,4)--(10,4);
    \draw (-0.45,1) node {$g$};
    \draw (-0.45,4) node {$\rho$};
    \draw (10.4,1) node {$g$};
    \draw (2,1) node {\large $\mathcal{R}$};
    \draw (4,0.2) node {$g^{\otimes k+1}$};
    \draw (6.76,2.5) node {\large $\Phi_U$};
    \draw (10.9,4) node {$U \rho U^\dagger$};
    \draw[decorate, decoration={brace,amplitude=5pt}] (0.7,5.5) -- (8.8,5.5) node[midway,yshift=12pt]{free operation};
\end{tikzpicture}
}
\vspace{-0.2cm}
\caption{Operational capabilities imposed on $g$. Resource replication generates $g^{\otimes k+1}$. Then $g^{\otimes k}$ is supplied to a free channel to simulate an arbitrary qubit unitary $U$. Thus, one copy of $g$ remains available after the simulated operation and can be reused catalytically.}
\label{fig:capabilities}
\end{figure}

Here, we remark that our setting is invariant under a common local unitary change of basis: conjugating the free-state hierarchy and the corresponding free operations by the same local unitary preserves all assumptions and both operational capabilities.
We therefore regard resource theories related by such a transformation as isomorphic.

\paragraph{Qubit resource-structure dichotomy.—}
We now specialize to qubit single systems. Let
$g=\ket{\gamma}\!\bra{\gamma}$ be the pure resource state and
$g^\perp=\ket{\gamma^\perp}\!\bra{\gamma^\perp}$ its orthogonal state.
We exclude the locally trivial cases in which the single-system free-state set is either a singleton or the entire Bloch ball, and call the remaining resource theories non-trivial. Operationally, a singleton free sector can neither carry nonzero classical information through the choice of free-state preparation nor admit any nontrivial variation of a free preparation, whereas the full Bloch ball leaves no state resource at all. To quantify the separation between the resource state and the free sector, define
\begin{equation}
    s_g:=\max_{\sigma\in\cF_1}\Tr(g\sigma).
    \label{eq:nontrivial-overlap}
\end{equation}
Since $g$ is nonfree, $s_g<1$; if $s_g=0$, every free qubit state is supported on $g^\perp$, forcing $\cF_1=\{g^\perp\}$. Hence $0<s_g<1$.

Within this nontrivial regime, the two operational capabilities impose a strong local rigidity.

\begin{lemma}[Overlap rigidity]
\label{lem:overlap-rigidity}
For a non-trivial resource theory with qubit single systems satisfying
Assumption~\ref{assumptions} and
Capabilities~\ref{cap:replication} and~\ref{cap:universal-simulation},
every free state $\sigma\in\cF_1$ satisfies
\begin{equation}
    \Tr(g\sigma)
    =
    \Tr(g^\perp\sigma)
    =
    \frac12.
    \label{eq:balance}
\end{equation}
\end{lemma}

The proof is given in the Supplemental Material.
Its essential ingredients are exact replication, a simulated projective measurement in the resource basis, and a simulated unitary interchanging $g$ and $g^\perp$.

Equation~\eqref{eq:balance} has a direct Bloch-sphere interpretation: every free state is unbiased with respect to the resource basis.
Choosing coordinates with $g=(I+Y)/2$, all free Bloch vectors therefore lie in the orthogonal $XZ$ plane, and the maximally mixed state is free.
By the non-triviality assumption, the local free geometry is therefore either one- or two-dimensional.
The Supplemental Material shows that these two possibilities complete, respectively, to the local structures of parity asymmetry and imaginarity.
Thus the operational requirements themselves generate a unitary--antiunitary dichotomy, even though no symmetry principle is assumed from the outset.

The consequences are not confined to single systems. In fact, the same
capabilities force a genuinely nonlocal structure of the composite free
sector. Define
\begin{equation}
    \FSEP(A{:}B)
    :=
    \operatorname{conv}
    \left\{
        \sigma_A\otimes\tau_B:
        \sigma_A\in\cF_A,\
        \tau_B\in\cF_B
    \right\}.
    \label{eq:fsep}
\end{equation}
With this notation, we show the following.

\begin{proposition}[Global free structure]
\label{prop:composite-structure}
For a non-trivial resource theory satisfying Assumption~\ref{assumptions} and \(s_g>0\),
Capability~\ref{cap:replication} alone implies
\begin{equation}
    \FSEP(A{:}B)
    \subsetneq
    \cF_{AB}.
    \label{eq:fsep-nogo}
\end{equation}
If, in addition, the single systems are qubits and
Capability~\ref{cap:universal-simulation} holds, then
$\cF_{AB}$ contains a maximally entangled pure state.
\end{proposition}

The proposition reveals two distinct levels of composite structure.
Already replication cannot be supported by a free sector constructed only from mixtures of locally free product states: copying the resource requires genuinely global free correlations.
This first conclusion does not rely on the qubit geometry or on universal simulation and therefore holds in arbitrary finite dimension under the stated condition \(s_g>0\).
For qubits, imposing universal utilization strengthens the constraint further, forcing maximal entanglement itself into the free sector.
The first statement is proved in the End Matter, while the second is established as part of the proof of Theorem~\ref{thm:classification} in the Supplemental Material.

What remains is whether the two local possibilities can be completed in different ways on larger composite systems.
They cannot: the alchemical requirements determine not only the local geometry but the entire finite-qubit resource theory.

\begin{theorem}[Complete finite-qubit classification of alchemical resource theories]
\label{thm:classification}
For a non-trivial alchemical resource theory with qubit single systems satisfying Assumption~\ref{assumptions}, the free-state hierarchy is, up to a common local unitary change of basis, exactly one of the following two structures for every finite \(n\).

\emph{(I) Imaginarity:}
\begin{equation}
    \cF_n
    =
    \left\{
        \rho\in\cD\!\left((\mathbb C^2)^{\otimes n}\right)
        :
        \rho^*=\rho
    \right\}.
    \label{eq:classification-imaginarity}
\end{equation}

\emph{(P) Parity asymmetry:}
\begin{equation}
    \cF_n
    =
    \left\{
        \rho\in\cD\!\left((\mathbb C^2)^{\otimes n}\right)
        :
        [\rho,Z^{\otimes n}]=0
    \right\},
    \label{eq:classification-parity}
\end{equation}
where $\cF_n$ denotes the free-state set of $n$ qubits.




\end{theorem}

The complete proof is given in the Supplemental Material.
The state classification also fixes the CFK dynamics: for a CFK Kraus operator $K:(\mathbb C^2)^{\otimes m}\to(\mathbb C^2)^{\otimes n}$, up to an overall phase, $K$ is real in case (I), while in case (P) it satisfies $Z^{\otimes n}K=\pm K Z^{\otimes m}$.
The theorem shows that the rigidity first visible in the single-qubit Bloch sphere propagates through the entire composite hierarchy:
neither alternative composite free-state structures nor alternative free dynamics remain compatible with the alchemical requirements.
The inverse problem requirements therefore recover, rather than assume, the unitary symmetry of parity asymmetry or the antiunitary symmetry of imaginarity.

\paragraph{A trade-off between exact universality and resource-state robustness.—} Capability~\ref{cap:universal-simulation} represents a maximal form of computational capability: together with the resource state $g$, free operations can exactly simulate an arbitrary qubit instrument, and hence an arbitrary qubit unitary.
Our classification reveals a sharp trade-off behind this universality.
Once the alchemical requirements are imposed simultaneously, resource states that are not exactly maximal lose their value for exact simulation of nonfree unitary dynamics.

This behavior is closely related to the exact no-programming theorem, which requires orthogonal program states for distinct unitary channels~\cite{NielsenChuang1997}, and to the requirement of perfect asymmetry for exact symmetry-breaking tasks~\cite{MarvianSpekkens2012}.
A closely related all-or-nothing phenomenon has also been established for imaginarity, where maximal imaginarity enables strict universality, whereas nonmaximal resources cannot implement non-real gates exactly~\cite{NakayamaTakeuchiAkibue2026}.

For the two theories in Theorem~\ref{thm:classification}, define
\begin{equation}
    \vartheta(\eta)
    =
    \begin{cases}
        \eta^*, & \text{(I)},\\
        Z\eta Z, & \text{(P)},
    \end{cases}
    \label{eq:simulation-involution}
\end{equation}
and call $\eta$ \emph{simulation-maximal} when
\begin{equation}
    F\!\left(\eta,\vartheta(\eta)\right)=0,
    \label{eq:simulation-max}
\end{equation}
where $F$ denotes the squared Uhlmann fidelity.
The induced CFK characterization above implies that a fixed free simulator implementing a unitary with resource state $\eta$ also implements the corresponding symmetry-transformed unitary with $\vartheta(\eta)$. The exact no-programming theorem then implies that, unless these two unitary channels coincide, the two resource states must be orthogonal; a detailed proof is given in the End Matter. This gives the same all-or-nothing structure in both surviving resource theories:
the designated maximal resource state $g$ can exactly simulate arbitrary qubit unitaries, whereas any departure from simulation maximality removes the ability to implement non-free unitary operations exactly.

\begin{corollary}[Exact-simulation trade-off]
\label{cor:simulation-tradeoff}
For a non-trivial resource theory satisfying Assumption~\ref{assumptions},
Capabilities~\ref{cap:replication} and
\ref{cap:universal-simulation} imply that any
non-simulation-maximal resource state can exactly simulate only free
unitary operations.
\end{corollary}

Thus the very requirements that guarantee maximal exact computational capability also eliminate any intermediate resource regime for exact unitary simulation.
In this sense, universality comes at the cost of extreme sensitivity to resource-state quality: a state is either simulation-maximal and capable of unlocking nonfree exact dynamics, or its resourcefulness has no operational value for exact unitary simulation beyond what is already achievable freely.

\paragraph{Conclusion.—}
We formulated the inverse problem of alchemical resource theory, asking which resource structures admit a resource that can be exactly replicated and can universally simulate quantum instruments.
Under natural consistency assumptions and branchwise complete freeness, we showed that, for qubit single systems with a non-trivial local free-state space, the only alchemical resource theories are imaginarity and parity asymmetry, up to a common local unitary change of basis.

We further showed that replication alone already requires genuinely global free correlations, while the full capabilities force maximally entangled states into the free sector.
At the same time, exact universality eliminates any intermediate resource regime: non-maximal resource states can exactly simulate only free unitary operations.
Our results therefore reveal a fundamental trade-off between computational capability and the precision required for physical implementation, imposing a new constraint on physically realizable quantum information processing.

An important open question is how far this rigidity can be relaxed.
Replacing exact replication or simulation by approximate transformations~\cite{HilleryZimanBuzek2006,YangRennerChiribella2020,HokkyoTajima2026}, or increasing the single-system dimension beyond qubits, may allow new nontrivial resource structures.
More broadly, one can abstract characteristic properties of other resource theories as operational capabilities and use them as inputs to the inverse problem.
Such extensions may clarify which capabilities rigidly determine a resource structure and which instead permit broader families of physically distinct theories.

\paragraph{Acknowledgments} Y.N.~and H.A. contributed equally to this work.
The authors used OpenAI's ChatGPT (GPT-5.5 and GPT-5.6 Sol) to assist with discussions and verification of mathematical proofs and with drafting and revising the manuscript.
The AI was used under detailed and iterative instructions specifying the mathematical arguments, logical structure, and intended presentation.
All AI-assisted outputs, including mathematical arguments and manuscript text, were thoroughly and independently checked, revised, and approved by the authors; the final proofs, scientific claims, and presentation were determined by the authors.
H.A.~was supported by JSPS KAKENHI Grant Nos.~25KJ0043, 26K17031, and 26K21941.

\clearpage
\onecolumngrid
\section*{End Matter}


\paragraph{RNG operations admit artificial realizations of the prescribed capabilities.---}
Resource-nongenerating (RNG) operations are widely used as the largest class of operations preserving free states, but can be too permissive to represent physically meaningful free operations~\cite{ChitambarGour2019,ChitambarGour2016,LiuHuLloyd2017,AlbarelliEtAl2018,SaxenaEtAl2020,TajimaTakagi2025}.
Here this permissiveness already shows that the two-system simulation condition identified in Remark~\ref{rem} as sufficient under CFK no longer controls the global free-state hierarchy when CFK is replaced by RNG.

Let
\[
g=\ket{i}\bra{i},
\qquad
\ket{i}=\frac{\ket0+i\ket1}{\sqrt2},
\qquad
\bar g=g^*,
\]
and define
\begin{equation}
\cF_n^{\mathrm{RNG}}
:=
\left\{
\rho:\rho^*=\rho,\ 
\Tr[g^{\otimes |S|}\rho_S]\le\frac14
\ \text{for every }S\subseteq[n]\text{ with }|S|\ge4
\right\}.
\label{eq:end-rng-free}
\end{equation}
These sets are closed, convex, permutation invariant, and marginally consistent.
They are also closed under products: if a tested subset contains at least four systems from one factor, the bound follows directly, while if it contains at most three systems from each factor, reality implies that each factor has overlap at most \(1/2\) with the corresponding power of \(g\), so the product overlap is at most \(1/4\).
Thus the construction satisfies the state-side Assumptions~(A1)--(A2), and its single-system free set is the non-trivial imaginarity disk; moreover \(g\notin\cF_1^{\mathrm{RNG}}\). The only change from the setting of the main theorem is that Assumption~(A3) is replaced by RNG operations.

For \(n\le3\), the additional constraint is absent and \(\cF_n^{\mathrm{RNG}}\) is exactly the imaginarity free-state set.
For \(n\ge4\), the hierarchy is different; for example,
\[
\frac12\bigl(g^{\otimes n}+\bar g^{\otimes n}\bigr)
\notin
\cF_n^{\mathrm{RNG}}.
\]
Nevertheless, exact replication is RNG through
\begin{equation}
\mathcal R(\rho)
=
\Tr(g\rho)\,g^{\otimes2}
+
\Tr(\bar g\rho)\,\bar g^{\otimes2}.
\label{eq:end-rng-rep}
\end{equation}
Indeed, every real qubit state has equal overlap \(1/2\) with \(g\) and \(\bar g\), so the output is real and hence free on two qubits.
Operationally, however, Eq.~\eqref{eq:end-rng-rep} measures the resource basis and conditionally prepares the resource states \(g^{\otimes2}\) and \(\bar g^{\otimes2}\). Thus the implementation explicitly uses resource-generating branches; RNG declares only the averaged channel free, because these branches compensate on free inputs. This is precisely the kind of hidden resource preparation that branchwise complete freeness is intended to exclude.

The same hierarchy satisfies the two-system simulation condition of Remark~\ref{rem}.
For a target instrument \(\{\mathcal N_j:A\to B\}_j\) with at most two-qubit input and output systems, use one program qubit in state \(g\).
Then \(A\otimes P\) contains at most three qubits and \(B\) at most two, so both the simulator input and output free-state sets coincide with those of standard imaginarity.
Writing \(\mathcal N_j(\rho)=\sum_\alpha M_{j\alpha}\rho M_{j\alpha}^\dagger\) and \(V_\pm:=I_A\otimes\ket{\pm i}\), define
\begin{align}
K_{j\alpha}^{(0)}
&=
\frac1{\sqrt2}
\left(
M_{j\alpha}V_+^\dagger
+
M_{j\alpha}^*V_-^\dagger
\right),
\nonumber\\
K_{j\alpha}^{(1)}
&=
\frac{i}{\sqrt2}
\left(
M_{j\alpha}V_+^\dagger
-
M_{j\alpha}^*V_-^\dagger
\right).
\label{eq:end-rng-sim}
\end{align}
These Kraus operators are real, form an instrument, and satisfy
\(\Phi_j(\rho\otimes g)=\mathcal N_j(\rho)\).
Hence every outcome map is RNG for the present hierarchy.
Therefore this theory satisfies local non-triviality, exact replication, and the two-system simulation condition actually sufficient for the CFK classification proof, in addition to all state-side consistency assumptions, yet its higher-system free-state hierarchy differs from both classified structures.
The obstruction is precisely that RNG checks freeness only on the nominal input system and does not require stability under arbitrary reference extensions, while CFK does.

\paragraph{Replication requires free correlations beyond locally free mixtures.---}
We prove here the first statement of Proposition~\ref{prop:composite-structure}.
Let $\sigma_*\in\cF_1$ attain $s_g$, so that
\begin{equation}
s_g=\Tr(g\sigma_*).
\end{equation}
For an exact free replicator $\mathcal R$ satisfying $\mathcal R(g)=g^{\otimes2}$, define
\begin{equation}
\Omega:=\mathcal R(\sigma_*)
\in
\cF_{AB}.
\end{equation}
By monotonicity of fidelity,
\begin{equation}
s_g
=
F(g,\sigma_*)
\le
F(g^{\otimes2},\Omega)
=
\Tr[(g\otimes g)\Omega].
\label{eq:end-fsep-fidelity}
\end{equation}

Suppose that $\Omega\in\FSEP(A{:}B)$.
Then
\begin{equation}
\Omega
=
\sum_x p_x\,\sigma_x\otimes\tau_x
\end{equation}
for locally free states $\sigma_x\in\cF_A$ and $\tau_x\in\cF_B$.
By the definition of $s_g$,
\begin{align}
\Tr[(g\otimes g)\Omega]
&=
\sum_x
p_x
\Tr(g\sigma_x)
\Tr(g\tau_x)
\nonumber\\
&\le
s_g^2.
\end{align}
Together with Eq.~\eqref{eq:end-fsep-fidelity}, this gives
\begin{equation}
s_g\le s_g^2,
\end{equation}
which contradicts \(0<s_g<1\), where \(s_g>0\) is assumed in Proposition~\ref{prop:composite-structure} and \(s_g<1\) follows from the nonfreeness of \(g\).
Therefore,
\begin{equation}
\FSEP(A{:}B)
\subsetneq
\cF_{AB}.
\end{equation}

This argument relies only on exact replication and the basic state-side consistency assumptions, and does not require the qubit geometry or universal simulation.
Thus genuinely global free correlations are already required by scalable resource supply itself.
As stated in the second part of Proposition~\ref{prop:composite-structure}, for qubits universal simulation strengthens this conclusion further by forcing a maximally entangled pure state itself to be free.

\paragraph{Proof of the exact-simulation trade-off.---}
We prove Corollary~\ref{cor:simulation-tradeoff}. Let a CFK processor exactly simulate a qubit unitary \(U\) with program state \(\eta\). By the classified CFK dynamics, the same processor with the transformed program \(\vartheta(\eta)\) of Eq.~\eqref{eq:simulation-involution} implements the transformed unitary
\[
U^\vartheta
=
\begin{cases}
U^*, & \text{case (I)},\\
ZUZ, & \text{case (P)}.
\end{cases}
\]
If \(\eta\) is not simulation-maximal, Eq.~\eqref{eq:simulation-max} gives \(F(\eta,\vartheta(\eta))>0\). The exact no-programming theorem therefore forces the two unitary channels to coincide~\cite{NielsenChuang1997}; for mixed programs, this follows by decomposing them into pure states and using extremality of unitary channels. Hence
\[
\mathcal U^\vartheta=\mathcal U,
\qquad\text{and therefore}\qquad
U^\vartheta=e^{i\phi}U .
\]

In case (I), \(U^*=e^{i\phi}U\), so \(U\) is real up to a global phase. In case (P), \(ZUZ=e^{i\phi}U\); applying the transformation twice gives \(e^{2i\phi}=1\), and thus
\[
ZU=\pm UZ .
\]
Therefore \(U\) is free in either case. Conversely, the designated resource \(g\) is simulation-maximal and exactly simulates arbitrary qubit unitaries by Capability~\ref{cap:universal-simulation}. This proves the all-or-nothing exact-simulation trade-off.

\clearpage
\onecolumngrid

\setcounter{section}{0}
\setcounter{equation}{0}
\setcounter{assumption}{0}
\setcounter{capability}{0}
\setcounter{proposition}{0}
\setcounter{lemma}{0}
\setcounter{theorem}{0}
\setcounter{corollary}{0}

\newcommand{\Herm}{\operatorname{Herm}}
\newcommand{\spanR}{\operatorname{span}_{\mathbb R}}
\newcommand{\id}{\operatorname{id}}

\begin{center}
{\large\bfseries Supplemental Material for ``Inverse Problem of Alchemical Resource Theory: Replication and Universal Simulation Single Out Imaginarity and Parity Asymmetry''}
\end{center}

\section{Organization and proof structure}

The proof of Theorem~\ref{thm:classification} is organized as a direct chain of lemmas.
We first establish the CFK compositional properties and then show that the two-system simulation condition in Remark~\ref{rem} is sufficient for Capability~\ref{cap:universal-simulation}.
The remaining lemmas then determine the free-state hierarchy.

Let
\[
S:=\spanR\{\sigma-I/2:\sigma\in\cF_1\}.
\]
The logical structure is as follows.

\begin{enumerate}
\item
Lemma~\ref{lem:SM-CFK-operation-closure} establishes reference extension and the sequential, tensor-product, and coarse-graining closure of CFK instruments.
Using these properties together with exact replication, Lemma~\ref{lem:SM-two-local-to-full} shows that the two-system simulation condition of Remark~\ref{rem} already implies Capability~\ref{cap:universal-simulation} on arbitrary finite composite systems.

\item
Lemma~\ref{lem:SM-replication-structure} derives the basic consequences of exact replication needed for the local analysis.
Combining them with the simulation capability, Lemma~\ref{lem:SM-overlap-rigidity} proves
\[
\Tr(g\sigma)
=
\Tr(g^\perp\sigma)
=
\frac12
\qquad
(\sigma\in\cF_1),
\]
and hence \(I/2\in\cF_1\).
Choosing \(g=(I+Y)/2\), this gives
\[
S\subseteq\spanR\{X,Z\}.
\]
The non-triviality assumption excludes \(S=\{0\}\), so \(\dim S=1\) or \(2\).
Accordingly, after a rotation about the \(Y\) axis, define
\[
\vartheta(H)
=
\begin{cases}
H^*, & \dim S=2,\\
ZHZ, & \dim S=1.
\end{cases}
\]

\item
For
\[
\mathsf E_m^\pm
:=
\{H:\vartheta^{\otimes m}(H)=\pm H\},
\]
Lemma~\ref{lem:SM-even-span-new} uses the free state
\[
\omega_+
=
\frac12
\left(
g^{\otimes2}
+
(g^\perp)^{\otimes2}
\right)
\]
together with product closure and suitable pure-state conversions to prove
\[
\mathsf E_m^+
\subseteq
\spanR\cF_m
\qquad
(\forall\,m).
\]

\item
Lemma~\ref{lem:SM-odd-elimination-new} uses Lemma~\ref{lem:SM-even-span-new}, complete free-cone preservation, and the finite-system simulation property from Lemma~\ref{lem:SM-two-local-to-full} to eliminate every anti-invariant component of a free state:
\[
\cF_m
\subseteq
\{\rho:\vartheta^{\otimes m}(\rho)=\rho\}.
\]
This is the upper inclusion in the classification.

\item
Lemma~\ref{lem:SM-lower-new} uses the upper inclusion just obtained, exact preparation of arbitrary pure states, and the free state \(\omega_+\) to prove the reverse inclusion
\[
\{\rho:\vartheta^{\otimes m}(\rho)=\rho\}
\subseteq
\cF_m.
\]
Hence, for every finite \(m\),
\[
\cF_m
=
\{\rho:\vartheta^{\otimes m}(\rho)=\rho\}.
\]
For \(\dim S=2\) this is imaginarity,
\[
\cF_m=\{\rho:\rho^*=\rho\},
\]
whereas for \(\dim S=1\) it is parity asymmetry,
\[
\cF_m=\{\rho:[\rho,Z^{\otimes m}]=0\}.
\]
This proves Theorem~\ref{thm:classification}.
\end{enumerate}

The dependencies can be summarized compactly as
\[
\begin{array}{c}
\text{Lemma~\ref{lem:SM-CFK-operation-closure}}
\ \Longrightarrow\
\text{Lemma~\ref{lem:SM-two-local-to-full}}
\\[1mm]
\text{Lemma~\ref{lem:SM-replication-structure}}
\ \Longrightarrow\
\text{Lemma~\ref{lem:SM-overlap-rigidity}}
\ \Longrightarrow\
\dim S\in\{1,2\},\ \vartheta
\\[1mm]
\Downarrow
\\[1mm]
\text{Lemma~\ref{lem:SM-even-span-new}}
\ \Longrightarrow\
\text{Lemma~\ref{lem:SM-odd-elimination-new}}
\ \Longrightarrow\
\cF_m\subseteq\Fix(\vartheta^{\otimes m})
\\[1mm]
\Downarrow
\\[1mm]
\text{Lemma~\ref{lem:SM-lower-new}}
\ \Longrightarrow\
\Fix(\vartheta^{\otimes m})\cap\cD
\subseteq\cF_m
\\[1mm]
\Downarrow
\\[1mm]
\cF_m=\Fix(\vartheta^{\otimes m})\cap\cD
\quad(\forall m)
\\[1mm]
\Downarrow
\\[1mm]
\text{Theorem~\ref{thm:classification}: imaginarity or parity asymmetry.}
\end{array}
\]

The characterization of CFK dynamics is logically downstream of this chain.
Proposition~\ref{prop:SM-induced-CFK-dynamics} characterizes the completely free Kraus operators after the state classification, and Proposition~\ref{prop:SM-ancilla-absorption} then absorbs free ancillas into direct Kraus representations.

\section{CFK framework and induced free dynamics}

We first formulate the CFK operation class used in the main text and establish the closure properties needed below.
For a single system considered without an explicit subsystem label, we write its free-state set as \(\cF_1\); for an explicitly labeled bipartite system \(AB\), the corresponding sets are \(\cF_A,\cF_B,\cF_{AB}\); and when only the number of systems matters, we write \(\cF_n\) for an \(n\)-qubit system.
The associated free cones are denoted analogously by \(\cC_1,\cC_A,\cC_B,\cC_{AB}\), or \(\cC_n\):
\begin{equation}
\cC_X := \left\{ t\sigma: t\ge0,\ \sigma\in\cF_X \right\}.
\label{eq:SM-free-cone}
\end{equation}

A linear operator \(K:A\to B\) is called a \emph{completely free Kraus operator} if, for every finite reference system \(R\),
\begin{equation}
(K\otimes I_R)\cC_{AR}(K^\dagger\otimes I_R)
\subseteq \cC_{BR}.
\label{eq:SM-CFK}
\end{equation}
Thus each such Kraus branch preserves the free cone even in the presence of arbitrary correlated free side information.

A completely positive map \(\Phi:A\to B\) is called a \emph{CFK operation} if there exist a finite ancillary system \(E\), a free state \(\eta_E\in\cF_E\), and a completely positive map
\[
\widetilde\Phi:AE\to B
\]
admitting a Kraus representation entirely by completely free Kraus operators such that
\begin{equation}
\Phi(X)=\widetilde\Phi(X\otimes\eta_E).
\label{eq:SM-CFK-operation}
\end{equation}
A CFK instrument \(\{\Phi_j:A\to B\}_j\) is defined similarly, with a \emph{common} finite free ancilla \(\eta_E\) and maps \(\widetilde\Phi_j:AE\to B\), each admitting completely free Kraus operators, such that
\[
\Phi_j(X)=\widetilde\Phi_j(X\otimes\eta_E)
\]
and \(\sum_j\Phi_j\) is trace preserving.
A one-outcome CFK instrument is a CFK channel.

\subsection{Basic closure properties}

\begin{lemma}[Closure of completely free Kraus operators]
\label{lem:SM-CFK-closure}
Completely free Kraus operators are closed under complex scalar multiplication, sequential composition, tensor-product composition, and operator-norm limits.
\end{lemma}

\begin{proof}
Scalar multiplication follows because \(\cC_{BR}\) is a cone.
For sequential composition, if \(K:A\to B\) and \(L:B\to C\) are completely free, then for arbitrary \(R\) and \(X\in\cC_{AR}\),
\[
(K\otimes I_R)X(K^\dagger\otimes I_R)\in\cC_{BR},
\]
and applying \(L\) with the same reference gives
\[
(LK\otimes I_R)X(K^\dagger L^\dagger\otimes I_R)\in\cC_{CR}.
\]
For tensor products, apply \(K\) while treating the second input and \(R\) as the reference, then apply \(L\) while treating the first output and \(R\) as the reference; permutation consistency restores the tensor ordering.
Finally, if \(K_n\to K\) in operator norm and every \(K_n\) is completely free, then
\[
(K_n\otimes I_R)X(K_n^\dagger\otimes I_R)
\longrightarrow
(K\otimes I_R)X(K^\dagger\otimes I_R),
\]
and closedness of the free cone gives the claim.
\end{proof}

\begin{lemma}[Basic properties of CFK operations]
\label{lem:SM-CFK-operation-closure}
\label{lem:SM-complete-preservation}
The CFK operation class satisfies the following properties.
\begin{enumerate}
\item \emph{Complete free-cone preservation.}
For every CFK operation \(\Phi:A\to B\), every finite reference system \(R\), and every \(X\in\cC_{AR}\),
\begin{equation}
(\Phi\otimes\mathrm{id}_R)(X)\in\cC_{BR}.
\label{eq:SM-complete-preservation}
\end{equation}
\item \emph{Reference extension.}
If \(\{\Phi_j:A\to B\}_j\) is a CFK instrument, then
\(\{\Phi_j\otimes\mathrm{id}_R:AR\to BR\}_j\) is CFK for every finite \(R\).
\item \emph{Sequential composition.}
The sequential composition of two CFK instruments is CFK, with the pair of outcomes retained.
\item \emph{Tensor-product composition.}
The tensor product of two CFK instruments is CFK.
\item \emph{Coarse-graining.}
Any classical coarse-graining of the outcomes of a CFK instrument is again a CFK instrument.
\item \emph{Free ancillary extension.}
For every \(\eta_E\in\cF_E\), the channel \(X\mapsto X\otimes\eta_E\) is CFK.
\end{enumerate}
\end{lemma}

\begin{proof}
Let
\[
\Phi_j(X)
=
\sum_\alpha
K_{j\alpha}(X\otimes\eta_E)K_{j\alpha}^\dagger
\]
be a common-ancilla realization of a CFK instrument, where every \(K_{j\alpha}:AE\to B\) is completely free.

For complete free-cone preservation, take \(X_{AR}\in\cC_{AR}\).
Product closure gives
\[
X_{AR}\otimes\eta_E\in\cC_{AER}.
\]
Applying each \(K_{j\alpha}\) while leaving \(R\) untouched gives an element of \(\cC_{BR}\); summing over \(\alpha\) remains in this convex cone.
This proves Eq.~\eqref{eq:SM-complete-preservation}.
The same realization with Kraus operators \(K_{j\alpha}\otimes I_R\) proves reference extension.

For sequential composition, let a second CFK instrument have free ancilla \(\zeta_F\) and completely free Kraus operators \(L_{k\beta}:BF\to C\).
Using the product free ancilla \(\eta_E\otimes\zeta_F\), the joint-outcome branch \((j,k)\) has Kraus operators
\[
L_{k\beta}(K_{j\alpha}\otimes I_F),
\]
which are completely free by Lemma~\ref{lem:SM-CFK-closure}.
The tensor-product statement follows analogously from the product free ancilla and Kraus operators
\[
K_{j\alpha}\otimes L_{k\beta}.
\]
Coarse-graining merely joins Kraus families belonging to the coarse-grained outcomes and therefore preserves the same common-ancilla realization.
Finally, \(X\mapsto X\otimes\eta_E\) is realized by adjoining \(\eta_E\) and applying the identity operator on \(AE\), which is completely free.
\end{proof}

\begin{lemma}[Two-system simulation suffices for finite-system simulation]
\label{lem:SM-two-local-to-full}
\label{lem:SM-one-copy}
Assume exact CFK replication and the two-system simulation condition of Remark~\ref{rem}: every instrument whose input and output each contain at most two single systems can be simulated exactly by a CFK instrument using finitely many copies of \(g\).
Then Capability~\ref{cap:universal-simulation} holds for arbitrary finite composite systems.
Moreover, exact replication reduces any such finite-copy simulation to an effective simulation from a single initial copy of \(g\).
\end{lemma}

\begin{proof}
Every finite-dimensional quantum instrument has an exact Stinespring--Naimark realization: one appends finitely many ancillary single systems in fixed pure states, applies a finite-dimensional unitary, performs projective measurements on ancillary systems, discards auxiliary systems, and classically coarse-grains the fine-grained measurement outcomes.
The trivial system is counted as a zero-system input or output.

Any unitary on finitely many single systems admits an exact decomposition into finitely many one- and two-system unitaries.
The required pure-state preparations, single-system measurements, and single-system discarding maps are themselves instruments with at most two input and output systems.
Hence every elementary step has an exact CFK simulator by the two-system condition of Remark~\ref{rem}.
A fixed pure ancillary state can, if necessary, be prepared by first adjoining any free single-system state and then simulating the corresponding single-system replacement channel.

Lemma~\ref{lem:SM-CFK-operation-closure} allows these local simulators to act on selected systems with all remaining systems treated as untouched references, and shows that their finite sequential and tensor-product compositions remain CFK.
The final classical coarse-graining is CFK by the same lemma.
Since the circuit contains only finitely many elementary steps, only finitely many copies of \(g\) are required.

Finally, repeated exact replication generates any prescribed finite number of resource copies from one initial copy of \(g\).
At each stage the already generated copies may be kept as untouched references by the reference-extension property in Lemma~\ref{lem:SM-CFK-operation-closure}.
Composing this replication stage with the finite-copy simulator yields an effective exact CFK simulation from one initial copy of \(g\).
\end{proof}

Hence the two-system condition of Remark~\ref{rem}, together with exact replication and CFK compositionality, suffices for the full Capability~\ref{cap:universal-simulation}; throughout the classification proof we may therefore use arbitrary finite-system simulation and, when convenient, its effective one-copy form.

\subsection{CFK dynamics induced by the classified free states}

Theorem~\ref{thm:classification} of the main text classifies the free-state hierarchy.
We first characterize the completely free Kraus operators associated with the two classified hierarchies.
We then show that, in these two theories, the free ancilla in Definition~\eqref{eq:SM-CFK-operation} can always be absorbed into the Kraus representation.
Thus the new CFK operation class reduces, after classification, to the usual bare branchwise operation classes.

For \(k\ge1\), let \(P_k:=Z^{\otimes k}\).

\begin{proposition}[Completely free Kraus operators induced by the classified free states]
\label{prop:SM-induced-CFK-dynamics}
Suppose that the free-state hierarchy is one of the two structures appearing in Theorem~\ref{thm:classification}.
\begin{enumerate}
\item \emph{Imaginarity.}
If
\[
\cF_k=\{\rho:\rho^*=\rho\}
\qquad
\text{for every }k\ge1,
\]
then \(K:(\mathbb C^2)^{\otimes m}\to(\mathbb C^2)^{\otimes n}\) is a completely free Kraus operator if and only if
\begin{equation}
K=e^{i\theta}R
\label{eq:SM-induced-imag-CFK}
\end{equation}
for some \(\theta\in\mathbb R\) and real matrix \(R\).

\item \emph{Parity asymmetry.}
If
\[
\cF_k=\{\rho:[\rho,P_k]=0\}
\qquad
\text{for every }k\ge1,
\]
then \(K:(\mathbb C^2)^{\otimes m}\to(\mathbb C^2)^{\otimes n}\) is a completely free Kraus operator if and only if
\begin{equation}
P_nK=\epsilon KP_m,
\qquad
\epsilon\in\{+1,-1\}.
\label{eq:SM-induced-parity-CFK}
\end{equation}
\end{enumerate}
\end{proposition}

\begin{proof}
\emph{Imaginarity.}
If \(K=e^{i\theta}R\) with \(R\) real, then for arbitrary reference \(Q\) and \(X\in\cC_{m+q}\),
\[
(K\otimes I_Q)X(K^\dagger\otimes I_Q)
=
(R\otimes I_Q)X(R^T\otimes I_Q)
\]
is positive and real, hence belongs to \(\cC_{n+q}\).
Thus \(K\) is completely free.

Conversely, let \(K\neq0\) be completely free, set \(d=2^m\), and take
\[
\ket{\Phi_d}
=
\frac1{\sqrt d}\sum_{x=0}^{d-1}\ket{x}\ket{x}.
\]
Since \(\ket{\Phi_d}\bra{\Phi_d}\in\cF_{2m}\),
\begin{equation}
(K\otimes I)\ket{\Phi_d}\bra{\Phi_d}(K^\dagger\otimes I)
=
\frac1d|K\rangle\!\rangle\langle\!\langle K|
\in\cC_{n+m}.
\label{eq:SM-induced-imag-Choi}
\end{equation}
The nonzero rank-one operator on the right is real, so its one-dimensional range is invariant under complex conjugation.
Hence \(e^{-i\theta}|K\rangle\!\rangle\) is real for some \(\theta\), and therefore \(e^{-i\theta}K\) is real.
The zero operator is trivial.

\emph{Parity asymmetry.}
If \(P_nK=\epsilon KP_m\), then for arbitrary reference \(Q\) and \(X\in\cC_{m+q}\),
\[
(P_n\otimes P_q)
(K\otimes I_Q)X(K^\dagger\otimes I_Q)
(P_n\otimes P_q)
=
(K\otimes I_Q)X(K^\dagger\otimes I_Q),
\]
so the output belongs to \(\cC_{n+q}\).

Conversely, for completely free \(K\neq0\), the same maximally entangled state satisfies
\[
(P_m\otimes P_m)\ket{\Phi_d}=\ket{\Phi_d},
\]
and therefore
\begin{equation}
\frac1d|K\rangle\!\rangle\langle\!\langle K|
\in\cC_{n+m}.
\label{eq:SM-induced-parity-Choi}
\end{equation}
This rank-one operator commutes with \(P_n\otimes P_m\), so
\[
(P_n\otimes P_m)|K\rangle\!\rangle
=
\epsilon|K\rangle\!\rangle
\]
for \(\epsilon=\pm1\).
Using
\[
(A\otimes B)|K\rangle\!\rangle=|AKB^T\rangle\!\rangle
\]
and \(P_m^T=P_m\) yields \(P_nK=\epsilon KP_m\).
\end{proof}

\begin{proposition}[Absorption of free ancillas in the classified theories]
\label{prop:SM-ancilla-absorption}
For both imaginarity and parity asymmetry, every CFK operation admits a Kraus representation consisting directly of completely free Kraus operators, without an explicit free ancilla.
The same holds simultaneously for all outcome maps of a CFK instrument.
\end{proposition}

\begin{proof}
Let a CFK instrument have common free ancilla
\[
\eta_E=\sum_\ell p_\ell\ket{f_\ell}\bra{f_\ell}
\]
and completely free dilation Kraus operators \(K_{j\alpha}:AE\to B\).

For imaginarity, \(\eta_E\) is real and positive, so its eigenvectors can be chosen real.
Thus every \(\ket{f_\ell}\bra{f_\ell}\) is a free pure state.
For parity asymmetry, \([\eta_E,P_E]=0\), so \(\eta_E\) can be diagonalized separately in the two parity sectors; every eigenvector can therefore be chosen with definite parity and is again a free pure state.

For each \(\ell\), the preparation operator
\[
V_\ell:=I_A\otimes\ket{f_\ell}:A\to AE
\]
is a completely free Kraus operator, because for any reference \(R\) and \(X\in\cC_{AR}\),
\[
(V_\ell\otimes I_R)X(V_\ell^\dagger\otimes I_R)
=
X\otimes\ket{f_\ell}\bra{f_\ell}
\in\cC_{AER}
\]
by product closure.
Hence
\[
L_{j\alpha\ell}
:=
\sqrt{p_\ell}\,K_{j\alpha}V_\ell
\]
is completely free by Lemma~\ref{lem:SM-CFK-closure}, and
\[
\Phi_j(X)
=
\sum_{\alpha,\ell}
L_{j\alpha\ell}X L_{j\alpha\ell}^\dagger.
\]
Thus the free ancilla is absorbed into a direct completely free Kraus representation.
\end{proof}

\begin{proposition}[Universal simulation in the classified theories]
\label{prop:SM-full-universal-simulation}
Both imaginarity and parity asymmetry satisfy exact simulation of arbitrary finite-system instruments with one resource copy.
\end{proposition}

\begin{proof}
Let \(\{\mathcal N_j:A\to B\}_j\) be an arbitrary instrument on finite-qubit systems, with Kraus operators
\[
\mathcal N_j(\rho)
=
\sum_\alpha M_{j\alpha}\rho M_{j\alpha}^\dagger,
\qquad
\sum_{j,\alpha}M_{j\alpha}^\dagger M_{j\alpha}=I_A.
\]
Choose the resource basis so that
\[
\ket{\gamma}
=
\frac{\ket0+i\ket1}{\sqrt2},
\qquad
\ket{\gamma^\perp}=\ket{\gamma}^*,
\]
and define
\[
V_+:=I_A\otimes\ket{\gamma},
\qquad
V_-:=I_A\otimes\ket{\gamma^\perp}.
\]

\emph{Imaginarity.}
For every \(j,\alpha\), define
\begin{align}
K_{j\alpha}^{(0)}
&:=
\frac1{\sqrt2}
\left(
M_{j\alpha}V_+^\dagger
+
M_{j\alpha}^*V_-^\dagger
\right),
\\
K_{j\alpha}^{(1)}
&:=
\frac{i}{\sqrt2}
\left(
M_{j\alpha}V_+^\dagger
-
M_{j\alpha}^*V_-^\dagger
\right).
\end{align}
Both are real and hence completely free by Proposition~\ref{prop:SM-induced-CFK-dynamics}.
Moreover,
\[
K_{j\alpha}^{(0)}V_+
=
\frac1{\sqrt2}M_{j\alpha},
\qquad
K_{j\alpha}^{(1)}V_+
=
\frac{i}{\sqrt2}M_{j\alpha},
\]
so the corresponding outcome map simulates \(\mathcal N_j\) exactly on resource input \(g\).
The cross terms cancel and
\[
\sum_{j,\alpha,s}
K_{j\alpha}^{(s)\dagger}K_{j\alpha}^{(s)}
=
I_{A\otimes P},
\]
so these Kraus operators form a bare CFK instrument, hence a CFK instrument.

\emph{Parity asymmetry.}
Let \(P_A\) and \(P_B\) denote the total parity operators on \(A\) and \(B\), and define
\[
L_{j\alpha}^{(\pm)}
:=
\frac1{\sqrt2}
\left(
M_{j\alpha}V_+^\dagger
\pm
P_BM_{j\alpha}P_AV_-^\dagger
\right).
\]
Since \(Z\ket{\gamma}=\ket{\gamma^\perp}\),
\[
P_BL_{j\alpha}^{(\pm)}
=
\pm
L_{j\alpha}^{(\pm)}(P_A\otimes Z),
\]
so these are completely free by Proposition~\ref{prop:SM-induced-CFK-dynamics}.
Also,
\[
L_{j\alpha}^{(+)}V_+
=
L_{j\alpha}^{(-)}V_+
=
\frac1{\sqrt2}M_{j\alpha},
\]
and
\[
\sum_{j,\alpha,\pm}
L_{j\alpha}^{(\pm)\dagger}L_{j\alpha}^{(\pm)}
=
I_{A\otimes P}.
\]
Thus they form a bare CFK instrument that exactly simulates the target instrument.
\end{proof}

By Proposition~\ref{prop:SM-ancilla-absorption}, the CFK dynamics of the two classified theories coincide with their usual branchwise descriptions: real Kraus representations for imaginarity and \(\mathbb Z_2\)-homogeneous Kraus representations \(P_nK=\pm KP_m\) for parity asymmetry.

\section{Overlap rigidity}

By Lemma~\ref{lem:SM-two-local-to-full}, the two operational capabilities assumed in the main text imply exact simulation of arbitrary finite-system instruments.
We therefore use this derived full simulation property throughout the classification proof.
All uses of a free operation on part of a larger free system are justified by the complete free-cone preservation and reference-extension properties in Lemma~\ref{lem:SM-CFK-operation-closure}.

Recall
\begin{equation}
s_g:=\max_{\sigma\in\cF_1}\Tr(g\sigma),
\qquad
m_g:=\min_{\sigma\in\cF_1}\Tr(g\sigma).
\label{eq:SM-minmax-overlap-new}
\end{equation}
The local non-triviality assumption in the main text implies $0<s_g<1$.
Since $\cF_1$ is compact, both extrema are attained.
Let
\[
g=\ket\gamma\bra\gamma,
\qquad
g^\perp=\ket{\gamma^\perp}\bra{\gamma^\perp}.
\]

\subsection{Replication-induced structure}

Let $\mathcal R:P\to AB$ be a fixed exact CFK replicator,
\begin{equation}
\mathcal R(g)=g\otimes g.
\label{eq:SM-replicator-new}
\end{equation}

\begin{lemma}[Replication-induced structure]
\label{lem:SM-replication-structure}
The following properties hold.
\begin{enumerate}
\item For every finite-system pure state $h$, there exists a CFK channel $\Lambda_h$ such that
\begin{equation}
\Lambda_h(g)=h.
\label{eq:SM-pure-conversion-new}
\end{equation}

\item The marginal channels $T_A:=\Tr_B\circ\mathcal R$ and $T_B:=\Tr_A\circ\mathcal R$ satisfy
\begin{equation}
\Tr[gT_A(\rho)]
=
\Tr[gT_B(\rho)]
=
\Tr(g\rho)
\label{eq:SM-marginal-invariance-new}
\end{equation}
for every qubit state $\rho$.

\item The same replicator also satisfies
\begin{equation}
\mathcal R(g^\perp)=g^\perp\otimes g^\perp.
\label{eq:SM-orthogonal-replication-new}
\end{equation}

\item For every $\sigma\in\cF_1$, writing $p_\sigma:=\Tr(g\sigma)$,
\begin{equation}
p_\sigma g+(1-p_\sigma)g^\perp\in\cF_1.
\label{eq:SM-dephasing-new}
\end{equation}
In particular, if $\sigma_{\max}$ and $\sigma_{\min}$ attain $s_g$ and $m_g$, then
\begin{equation}
\tau_{\max}:=s_g g+(1-s_g)g^\perp\in\cF_1,
\qquad
\tau_{\min}:=m_g g+(1-m_g)g^\perp\in\cF_1,
\label{eq:SM-tau-extrema-new}
\end{equation}
and
\begin{equation}
\Omega_{\max}
:=
s_g g\otimes g+(1-s_g)g^\perp\otimes g^\perp
\in\cF_2.
\label{eq:SM-Omega-max-new}
\end{equation}
\end{enumerate}
\end{lemma}

\begin{proof}
For the first statement, fix any pure $h$ and consider the replacement channel
$\mathcal N_h(\rho)=\Tr(\rho)h$.
By effective one-copy simulation, there is a CFK channel
$\Psi_h:A\otimes P\to B$ satisfying
$\Psi_h(\rho\otimes g)=h$ for every state $\rho$.
Choose any free state $\sigma_0\in\cF_A$ and define
$\Lambda_h(X):=\Psi_h(\sigma_0\otimes X)$.
Since adjoining a free ancilla is included in the CFK definition, $\Lambda_h$ is CFK and $\Lambda_h(g)=h$.

For the second statement, $T_A(g)=T_B(g)=g$.
For $X=A,B$, put $E_X:=T_X^\dagger(g)$.
Since $0\le E_X\le I$ and $\Tr(E_Xg)=1$, positivity gives
\[
E_X=g+r_Xg^\perp,
\qquad
0\le r_X\le1.
\]
Let $\sigma_{\max}\in\cF_1$ attain $s_g$.
Because $\mathcal R(\sigma_{\max})$ is free and marginals of free states are free,
\[
\Tr[gT_X(\sigma_{\max})]
=s_g+r_X(1-s_g)
\le s_g.
\]
Since $s_g<1$, $r_X=0$, and Eq.~\eqref{eq:SM-marginal-invariance-new} follows.

Applying Eq.~\eqref{eq:SM-marginal-invariance-new} to $g^\perp$ shows that both marginals of
$\mathcal R(g^\perp)$ equal the pure state $g^\perp$; hence Eq.~\eqref{eq:SM-orthogonal-replication-new} follows.

Finally, let $\Omega_\sigma:=\mathcal R(\sigma)$ for $\sigma\in\cF_1$.
Fidelity monotonicity and Eq.~\eqref{eq:SM-marginal-invariance-new} give
\[
p_\sigma
=F(g,\sigma)
\le
\Tr[(g\otimes g)\Omega_\sigma]
\le
\Tr[(g\otimes I)\Omega_\sigma]
=p_\sigma.
\]
Thus equality holds throughout.
In the product resource basis, positivity then forces the $g\otimes g^\perp$ and $g^\perp\otimes g$ rows and columns of $\Omega_\sigma$ to vanish.
Consequently both marginals are
$p_\sigma g+(1-p_\sigma)g^\perp$, which is free by marginal consistency.
Applying this to $\sigma_{\max}$ and $\sigma_{\min}$ gives Eq.~\eqref{eq:SM-tau-extrema-new}, while linearity of $\mathcal R$ and Eq.~\eqref{eq:SM-orthogonal-replication-new} give Eq.~\eqref{eq:SM-Omega-max-new}.
\end{proof}

\subsection{Proof of overlap rigidity}

\begin{lemma}[Overlap rigidity]
\label{lem:SM-overlap-rigidity}
Every free qubit state $\sigma\in\cF_1$ satisfies
\begin{equation}
\Tr(g\sigma)=\Tr(g^\perp\sigma)=\frac12,
\label{eq:SM-overlap-rigidity}
\end{equation}
and
\begin{equation}
\frac I2\in\cF_1.
\label{eq:SM-maxmix-free-new}
\end{equation}
\end{lemma}

\begin{proof}
We first show $m_g>0$.
Assume instead $m_g=0$.
Equation~\eqref{eq:SM-dephasing-new} then gives $g^\perp\in\cF_1$.
Consider the L\"uders instrument in the resource basis,
\[
\mathcal N_0(\rho)=g\rho g,
\qquad
\mathcal N_1(\rho)=g^\perp\rho g^\perp,
\]
and let $\{\Psi_0,\Psi_1\}$ be a one-copy CFK simulation.

Apply $\mathcal R$ to one subsystem of the free state $\Omega_{\max}$ in Eq.~\eqref{eq:SM-Omega-max-new}.
By Lemma~\ref{lem:SM-complete-preservation}, with the other subsystem kept as a reference, the state
\[
\Gamma_{\max}
=
s_g g_A\otimes g_P\otimes g_R
+(1-s_g)g_A^\perp\otimes g_P^\perp\otimes g_R^\perp
\]
is free.
Let
\[
X:=\Psi_0(g^\perp\otimes g^\perp),
\qquad
q:=\Tr X\le1.
\]
Applying $\Psi_0$ to $AP$ and leaving $R$ untouched gives, again by Lemma~\ref{lem:SM-complete-preservation},
\[
s_g g_B\otimes g_R
+(1-s_g)X_B\otimes g_R^\perp
\in\cC_{BR}.
\]
Taking the $R$ marginal and normalizing, the maximality of $s_g$ yields
\[
\frac{s_g}{s_g+(1-s_g)q}\le s_g.
\]
Since $0<s_g<1$ and $q\le1$, this forces $q=1$.
As $\Psi_0+\Psi_1$ is trace preserving,
$\Psi_1(g^\perp\otimes g^\perp)=0$.

Now $g_A^\perp\otimes\Omega_{\max}^{PR}$ is free by product closure.
Applying $\Psi_1$ to $AP$ while keeping $R$ as reference gives
\[
s_g g_B^\perp\otimes g_R\in\cC_{BR},
\]
where exact simulation was used on the program-$g$ branch and the program-$g^\perp$ branch vanishes by the preceding paragraph.
Taking the $B$ marginal gives $s_g g_R\in\cC_R$, hence $g\in\cF_1$, a contradiction.
Therefore $m_g>0$.

We next show $m_g=s_g$.
Let $\Phi_0$ be the $g$-outcome map of a one-copy simulation of the same L\"uders instrument, and choose an arbitrary Kraus representation
$\Phi_0(X)=\sum_\alpha K_\alpha X K_\alpha^\dagger$.
Let
\[
J_g\ket\varphi=\ket\varphi\otimes\ket\gamma.
\]
Because the restriction of $\Phi_0$ to the program-$g$ subspace is the rank-one map $\rho\mapsto g\rho g$, every Kraus operator satisfies
$K_\alpha J_g=c_\alpha g$ with $\sum_\alpha|c_\alpha|^2=1$.
In the ordered product resource basis, write
\[
K_\alpha=
\begin{pmatrix}
c_\alpha&a_\alpha&0&u_\alpha\\
0&b_\alpha&0&v_\alpha
\end{pmatrix}.
\]
Define
\[
Y:=\sum_\alpha
\binom{a_\alpha}{b_\alpha}
(\bar a_\alpha\ \bar b_\alpha),
\qquad
Z:=\sum_\alpha
\binom{u_\alpha}{v_\alpha}
(\bar u_\alpha\ \bar v_\alpha).
\]
Trace nonincrease gives $\Tr Y,\Tr Z\le1$.
For every $W\in\cC_1$,
\begin{equation}
m_g\Tr W\le\Tr(gW)\le s_g\Tr W.
\label{eq:SM-overlap-bounds-new}
\end{equation}

Applying $\Phi_0$ to $\Omega_{\max}$ gives the free-cone element
$s_g g+(1-s_g)Z$.
The upper bound in Eq.~\eqref{eq:SM-overlap-bounds-new} yields
\[
Z_{00}\le s_g(\Tr Z-1).
\]
The left-hand side is nonnegative and the right-hand side nonpositive, hence
$Z_{00}=0$ and $\Tr Z=1$.
Positivity of $Z$ therefore gives
\begin{equation}
Z=g^\perp.
\label{eq:SM-Z-new}
\end{equation}

For $\tau_x:=xg+(1-x)g^\perp$,
\[
\Phi_0(\tau_x^{\otimes2})
=
x^2g+x(1-x)Y+(1-x)^2g^\perp.
\]
Because $\tau_{\max}^{\otimes2}$ and $\tau_{\min}^{\otimes2}$ are free, applying respectively the upper and lower bounds in Eq.~\eqref{eq:SM-overlap-bounds-new} gives
\[
Y_{00}\le s_g\Tr Y+1-2s_g,
\qquad
Y_{00}\ge m_g\Tr Y+1-2m_g.
\]
Hence
\[
(s_g-m_g)(\Tr Y-2)\ge0.
\]
Since $s_g\ge m_g$ and $\Tr Y\le1$, we conclude
\begin{equation}
m_g=s_g.
\label{eq:SM-minmax-equal-new}
\end{equation}
Thus every free state has the same overlap $s_g$ with $g$.

Finally, let
\[
U=\ket{\gamma^\perp}\bra\gamma+\ket\gamma\bra{\gamma^\perp}
\]
and let $\Phi_U$ be a one-copy CFK simulation of $\operatorname{Ad}_U$.
Choose any Kraus representation $\Phi_U(X)=\sum_\alpha L_\alpha XL_\alpha^\dagger$.
Exact simulation on the program-$g$ subspace implies
$L_\alpha J_g=c_\alpha U$ and $\sum_\alpha|c_\alpha|^2=1$.
Thus
\[
L_\alpha=
\begin{pmatrix}
0&a_\alpha&c_\alpha&u_\alpha\\
c_\alpha&b_\alpha&0&v_\alpha
\end{pmatrix}.
\]
Define $Y$ and $Z$ from the second and fourth columns as above.
Trace preservation gives $Y,Z\ge0$ and $\Tr Y=\Tr Z=1$.

Applying $\Phi_U$ to $\Omega_{\max}$ gives the normalized free state
$s_gg^\perp+(1-s_g)Z$.
Its overlap with $g$ must equal $s_g$, hence
\[
Z_{00}=\frac{s_g}{1-s_g}\le1,
\]
so $s_g\le1/2$.
Applying $\Phi_U$ to $\tau_{\max}^{\otimes2}$ and again using the fixed overlap gives
\[
Y_{00}=\frac{2s_g-1}{1-s_g}\ge0,
\]
so $s_g\ge1/2$.
Therefore
\[
m_g=s_g=\frac12,
\]
which proves Eq.~\eqref{eq:SM-overlap-rigidity}.
Equation~\eqref{eq:SM-dephasing-new} then gives $I/2\in\cF_1$.
\end{proof}

\section{Complete finite-qubit classification}

We now prove Theorem~\ref{thm:classification} of the main text.
Choose a common local basis such that
\begin{equation}
g=\frac12(I+Y),
\qquad
g^\perp=\frac12(I-Y).
\label{eq:SM-gY-new}
\end{equation}
By Lemma~\ref{lem:SM-overlap-rigidity}, every free qubit state has vanishing $Y$ component and $I/2$ is free.
Define
\begin{equation}
S:=\spanR\{\sigma-I/2:\sigma\in\cF_1\}
\subseteq\spanR\{X,Z\}.
\label{eq:SM-local-S-new}
\end{equation}
Because the local free-state space is non-trivial, $S\neq\{0\}$, and hence
\begin{equation}
\dim S\in\{1,2\}.
\label{eq:SM-local-dim-new}
\end{equation}

If $\dim S=2$, define the involution
\begin{equation}
\vartheta(H):=H^*.
\label{eq:SM-theta-imag-new}
\end{equation}
If $\dim S=1$, a rotation about the $Y$ axis leaves $g$ unchanged and allows us to take $S=\spanR\{Z\}$; define
\begin{equation}
\vartheta(H):=ZHZ.
\label{eq:SM-theta-parity-new}
\end{equation}
For $m\ge1$, let
\begin{equation}
\mathsf E_m^\pm
:=
\{H\in\Herm((\mathbb C^2)^{\otimes m}):
\vartheta^{\otimes m}(H)=\pm H\}.
\label{eq:SM-even-odd-new}
\end{equation}
Thus \(\mathsf E_m^+\) and \(\mathsf E_m^-\) are respectively the invariant and anti-invariant subspaces of \(\vartheta^{\otimes m}\).
In the first case, $\mathsf E_m^+$ is spanned by Pauli strings containing an even number of $Y$ factors; in the second, it is spanned by strings containing an even number of factors from $\{X,Y\}$.

\subsection{The invariant operator space lies in the free-state span}

\begin{lemma}[Invariant-span lemma]
\label{lem:SM-even-span-new}
For every $m\ge1$,
\begin{equation}
\mathsf E_m^+
\subseteq
\spanR\cF_m.
\label{eq:SM-even-span-new}
\end{equation}
\end{lemma}

\begin{proof}
First suppose $\dim S=2$.
Then
\[
\spanR\cF_1=\spanR\{I,X,Z\}.
\]
By Lemma~\ref{lem:SM-replication-structure} and $I/2\in\cF_1$,
\begin{equation}
\omega_+
:=
\mathcal R(I/2)
=
\frac12(g\otimes g+g^\perp\otimes g^\perp)
=
\frac14(I\otimes I+Y\otimes Y)
\in\cF_2.
\label{eq:SM-omega-plus-new}
\end{equation}
Since $I\otimes I\in\spanR\cF_2$, this gives
$Y\otimes Y\in\spanR\cF_2$.
Product closure implies that tensor products of operators from free-state spans again belong to the corresponding free-state span.
Every Pauli string with an even number of $Y$ factors can therefore be built from single-site factors $I,X,Z$ and paired $Y$ factors.
Hence Eq.~\eqref{eq:SM-even-span-new} holds.

Now suppose $\dim S=1$, with $S=\spanR\{Z\}$.
Then
\[
\spanR\cF_1=\spanR\{I,Z\}.
\]
For $A\in\{X,Y\}$, put
\[
\psi_A:=\frac12(I+A).
\]
By Lemma~\ref{lem:SM-replication-structure}, choose a CFK channel $\Lambda_A$ satisfying
$\Lambda_A(g)=\psi_A$.
Since $\Lambda_A(I/2)$ is free, there is $t_A\in[-1,1]$ such that
\[
\Lambda_A(I/2)=\frac12(I+t_AZ).
\]
Linearity gives
\[
\Lambda_A(g^\perp)
=2\Lambda_A(I/2)-\Lambda_A(g)
=
\frac12(I-A+2t_AZ).
\]
This is a density operator, so its Bloch vector must have norm at most one.
Because $A\perp Z$ and $\|A\|=1$, this implies
$1+4t_A^2\le1$, hence
\begin{equation}
\Lambda_A(I/2)=I/2,
\qquad
\Lambda_A(g^\perp)=\frac12(I-A).
\label{eq:SM-parity-conversion-new}
\end{equation}
For $A,B\in\{X,Y\}$, tensor closure of CFK operations and Eq.~\eqref{eq:SM-omega-plus-new} give
\begin{equation}
(\Lambda_A\otimes\Lambda_B)(\omega_+)
=
\frac14(I\otimes I+A\otimes B)
\in\cF_2.
\label{eq:SM-parity-pairs-new}
\end{equation}
Thus every pair of local anti-invariant factors $A,B\in\{X,Y\}$ lies in $\spanR\cF_2$.
Together with the local invariant factors $I,Z$, product closure generates every Pauli string with an even number of $X/Y$ factors.
This proves Eq.~\eqref{eq:SM-even-span-new} also in the one-dimensional case.
\end{proof}

\subsection{Universal simulation eliminates the anti-invariant sector}

\begin{lemma}[Anti-invariant-component elimination]
\label{lem:SM-odd-elimination-new}
For every $m\ge1$,
\begin{equation}
\cF_m
\subseteq
\{\rho:\vartheta^{\otimes m}(\rho)=\rho\}.
\label{eq:SM-upper-new}
\end{equation}
\end{lemma}

\begin{proof}
Let $C$ be any Pauli string in $\mathsf E_m^-$; then $C=C^\dagger$ and $C^2=I$.
Define
\[
P_\pm:=\frac12(I\pm C)
\]
and the channel
\begin{equation}
\mathcal N_C(\rho)
:=
\Tr(P_+\rho)g+\Tr(P_-\rho)g^\perp.
\label{eq:SM-NC-new}
\end{equation}
Using $Y=g-g^\perp$ gives
\begin{equation}
\mathcal N_C^\dagger(Y)=C.
\label{eq:SM-NC-adjoint-new}
\end{equation}
By Capability~\ref{cap:universal-simulation} and the effective one-copy form established in Lemma~\ref{lem:SM-two-local-to-full}, there exists a CFK channel
$\Phi_C:m+1\to1$ satisfying
\[
\Phi_C(\rho\otimes g)=\mathcal N_C(\rho).
\]
Set
\[
H_C:=\Phi_C^\dagger(Y).
\]
Since $\Phi_C^\dagger$ is unital and completely positive,
$-I\le H_C\le I$.
Let $J_g\ket\psi=\ket\psi\otimes\ket\gamma$.
Equation~\eqref{eq:SM-NC-adjoint-new} gives
\begin{equation}
J_g^\dagger H_CJ_g=C.
\label{eq:SM-compression-new}
\end{equation}
Because $C^2=I$, the compression in Eq.~\eqref{eq:SM-compression-new} saturates the norm bound on each $\pm1$ eigenspace of $C$.
For a $+1$ eigenvector $v$, positivity of $I-H_C$ and
$\langle v,g|(I-H_C)|v,g\rangle=0$ imply
$H_C\ket{v,g}=\ket{v,g}$; similarly, a $-1$ eigenvector satisfies
$H_C\ket{v,g}=-\ket{v,g}$.
Thus the program-$g$ subspace reduces $H_C$, and
\begin{equation}
H_C=C\otimes g+K_C\otimes g^\perp
\label{eq:SM-H-block-new}
\end{equation}
for some Hermitian contraction $K_C$ on the $m$ data qubits.

Let $\Omega\in\cF_m$.
Since $I/2$ is free, $\Omega\otimes I/2$ is free.
By CFK freeness, $\Phi_C(\Omega\otimes I/2)$ is a free qubit state, whose $Y$ expectation vanishes by Lemma~\ref{lem:SM-overlap-rigidity}.
Hence
\begin{equation}
0
=
\Tr[H_C(\Omega\otimes I/2)]
=
\frac12\Tr[(C+K_C)\Omega].
\label{eq:SM-CplusK-new}
\end{equation}
Therefore $C+K_C$ is orthogonal to $\spanR\cF_m$.
By Lemma~\ref{lem:SM-even-span-new}, it is orthogonal to $\mathsf E_m^+$.
Since $C\in\mathsf E_m^-$, the invariant component of $K_C$ must vanish; hence
\begin{equation}
K_C\in\mathsf E_m^-.
\label{eq:SM-K-odd-new}
\end{equation}

We now show $K_C=C$.
Let $\mathcal O_1$ denote a basis of the local anti-invariant space:
\[
\mathcal O_1=\{Y\}
\quad\text{in the imaginarity case},
\qquad
\mathcal O_1=\{X,Y\}
\quad\text{in the parity case}.
\]
For $A\in\mathcal O_1$, define
\begin{equation}
\chi_A:=\frac14(I\otimes I+A\otimes Y).
\label{eq:SM-chi-new}
\end{equation}
In the imaginarity case, $\chi_Y=\omega_+$.
In the parity case, Eq.~\eqref{eq:SM-parity-pairs-new} with $B=Y$ shows $\chi_A\in\cF_2$.

Fix a data site $j$ and $E\in\mathsf E_{m-1}^+$.
By Lemma~\ref{lem:SM-even-span-new}, product closure, and permutation consistency,
$\chi_A^{jP}\otimes E$ belongs to $\spanR\cF_{m+1}$.
The linear functional
$X\mapsto\Tr[H_CX]$ vanishes on $\spanR\cF_{m+1}$ because a CFK output has zero $Y$ expectation.
Using
\[
H_C
=
\frac12[(C+K_C)\otimes I+(C-K_C)\otimes Y]
\]
and Eq.~\eqref{eq:SM-chi-new}, we obtain
\[
0
\propto
\Tr[(C+K_C)(I_j\otimes E)]
+
\Tr[(C-K_C)(A_j\otimes E)].
\]
The first term vanishes because $C+K_C\in\mathsf E_m^-$ and $I_j\otimes E\in\mathsf E_m^+$.
Thus
\begin{equation}
\Tr[(C-K_C)(A_j\otimes E)]=0.
\label{eq:SM-CK-orth-new}
\end{equation}
As $j$, $A\in\mathcal O_1$, and $E\in\mathsf E_{m-1}^+$ vary, the operators $A_j\otimes E$ span $\mathsf E_m^-$: every Pauli string in the anti-invariant subspace contains at least one local anti-invariant factor, and removing that factor leaves an invariant string.
By Eq.~\eqref{eq:SM-K-odd-new}, $C-K_C\in\mathsf E_m^-$, so Eq.~\eqref{eq:SM-CK-orth-new} implies
\begin{equation}
K_C=C.
\label{eq:SM-KequalsC-new}
\end{equation}
Returning to Eq.~\eqref{eq:SM-CplusK-new} gives
\[
\Tr(C\Omega)=0
\qquad
\forall\,\Omega\in\cF_m.
\]
The Pauli strings in $\mathsf E_m^-$ form a basis of the anti-invariant subspace, so every free state has vanishing anti-invariant component.
This proves Eq.~\eqref{eq:SM-upper-new}.
\end{proof}

\subsection{Strong state universality gives the lower inclusion}

\begin{lemma}[Lower inclusion]
\label{lem:SM-lower-new}
For every $m\ge1$,
\begin{equation}
\{\rho:\vartheta^{\otimes m}(\rho)=\rho\}
\subseteq
\cF_m.
\label{eq:SM-lower-new}
\end{equation}
\end{lemma}

\begin{proof}
Let $h$ be an arbitrary pure $m$-qubit state.
By Lemma~\ref{lem:SM-replication-structure}, choose a CFK channel
$\Lambda_h:1\to m$ such that
$\Lambda_h(g)=h$.
Set
\[
\eta_h:=\Lambda_h(g^\perp).
\]
Applying $\Lambda_h$ to one subsystem of the free state $\omega_+$ in Eq.~\eqref{eq:SM-omega-plus-new}, while leaving the other subsystem untouched, gives by Lemma~\ref{lem:SM-complete-preservation}
\begin{equation}
\Omega_h
:=
\frac12(h\otimes g+\eta_h\otimes g^\perp)
\in\cF_{m+1}.
\label{eq:SM-Omega-h-new}
\end{equation}
By Lemma~\ref{lem:SM-odd-elimination-new}, $\Omega_h$ is fixed by
$\vartheta^{\otimes m}\otimes\vartheta$.
In both cases defined above,
\[
\vartheta(g)=g^\perp,
\qquad
\vartheta(g^\perp)=g.
\]
Comparing the two orthogonal program blocks in Eq.~\eqref{eq:SM-Omega-h-new} therefore gives
\begin{equation}
\eta_h=\vartheta^{\otimes m}(h).
\label{eq:SM-partner-new}
\end{equation}
Since $I/2=(g+g^\perp)/2$ is free,
\begin{equation}
\Lambda_h(I/2)
=
\frac12[h+\vartheta^{\otimes m}(h)]
\in\cF_m.
\label{eq:SM-symmetrized-pure-new}
\end{equation}

Now let $\rho$ be any state fixed by $\vartheta^{\otimes m}$ and write a spectral decomposition
$\rho=\sum_i p_ih_i$ into pure states.
Then
\[
\rho
=
\frac12[\rho+\vartheta^{\otimes m}(\rho)]
=
\sum_i p_i\frac12[h_i+\vartheta^{\otimes m}(h_i)].
\]
Every term is free by Eq.~\eqref{eq:SM-symmetrized-pure-new}, and convexity gives $\rho\in\cF_m$.
This proves Eq.~\eqref{eq:SM-lower-new}.
\end{proof}

Combining Lemmas~\ref{lem:SM-odd-elimination-new} and~\ref{lem:SM-lower-new}, we obtain for every finite $m$:
\begin{equation}
\cF_m
=
\{\rho:\vartheta^{\otimes m}(\rho)=\rho\}.
\label{eq:SM-final-classification-new}
\end{equation}
If $\dim S=2$, this is
\[
\cF_m=\{\rho:\rho^*=\rho\},
\]
namely imaginarity.
If $\dim S=1$, this is
\[
\cF_m=\{\rho:[\rho,Z^{\otimes m}]=0\},
\]
namely parity asymmetry.
This proves Theorem~\ref{thm:classification} of the main text.

As an immediate consequence, the Bell state
\begin{equation}
\ket{\Phi^+}
=
\frac{\ket{00}+\ket{11}}{\sqrt2}
\label{eq:SM-Bell-corollary-new}
\end{equation}
is free in both cases: its projector is real and has even parity.
Thus the second statement of Proposition~\ref{prop:composite-structure} in the main text also follows.






\begin{thebibliography}{99}

\bibitem{ChitambarGour2019}
E.~Chitambar and G.~Gour,
Quantum resource theories,
Rev. Mod. Phys. \textbf{91}, 025001 (2019).

\bibitem{BrandaoGour2015}
F.~G.~S.~L.~Brand\~{a}o and G.~Gour,
Reversible framework for quantum resource theories,
Phys. Rev. Lett. \textbf{115}, 070503 (2015).

\bibitem{Gour2017}
G.~Gour,
Quantum resource theories in the single-shot regime,
Phys. Rev. A \textbf{95}, 062314 (2017).

\bibitem{Regula2018}
B.~Regula,
Convex geometry of quantum resource quantification,
J. Phys. A \textbf{51}, 045303 (2018).

\bibitem{HorodeckiEtAl2009}
R.~Horodecki,
P.~Horodecki,
M.~Horodecki, and K.~Horodecki,
Quantum entanglement,
Rev. Mod. Phys. \textbf{81}, 865 (2009).

\bibitem{BaumgratzEtAl2014}
T.~Baumgratz,
M.~Cramer, and M.~B.~Plenio,
Quantifying coherence,
Phys. Rev. Lett. \textbf{113}, 140401 (2014).

\bibitem{GourSpekkens2008}
G.~Gour and R.~W.~Spekkens,
The resource theory of quantum reference frames: Manipulations and monotones,
New J. Phys. \textbf{10}, 033023 (2008).

\bibitem{BrandaoEtAl2013}
F.~G.~S.~L.~Brand\~{a}o,
M.~Horodecki,
J.~Oppenheim,
J.~M.~Renes, and R.~W.~Spekkens,
Resource theory of quantum states out of thermal equilibrium,
Phys. Rev. Lett. \textbf{111}, 250404 (2013).

\bibitem{VeitchEtAl2014}
V.~Veitch,
S.~A.~Hamed Mousavian,
D.~Gottesman, and J.~Emerson,
The resource theory of stabilizer quantum computation,
New J. Phys. \textbf{16}, 013009 (2014).

\bibitem{ScandiSurace2021}
M.~Scandi and J.~Surace,
Undecidability in resource theory: Can you tell resource theories apart?,
Phys. Rev. Lett. \textbf{127}, 270501 (2021).

\bibitem{NagasawaEtAl2025}
T.~Nagasawa,
E.~Wakakuwa,
K.~Kato, and F.~Buscemi,
Macroscopicity and observational deficit in states, operations, and correlations,
Rep. Prog. Phys. \textbf{88}, 117601 (2025).

\bibitem{LieEtAl2026}
S.~H.~Lie,
J.~Son,
P.~Boes,
N.~H.~Y.~Ng, and H.~Wilming,
Thermal operations from informational equilibrium,
Phys. Rev. Lett. \textbf{137}, 030403 (2026).

\bibitem{SonEtAl2026}
J.~Son,
R.~Ganardi,
S.~Minagawa,
F.~Buscemi,
S.~H.~Lie, and N.~H.~Y.~Ng,
Catalytic channels are the only noise-robust catalytic processes,
Phys. Rev. Lett. \textbf{136}, 050202 (2026).

\bibitem{KuroiwaYamasaki2020}
K.~Kuroiwa and H.~Yamasaki,
General quantum resource theories: Distillation, formation and consistent resource measures,
Quantum \textbf{4}, 355 (2020).

\bibitem{BartlettRudolphSpekkens2007}
S.~D.~Bartlett,
T.~Rudolph, and R.~W.~Spekkens,
Reference frames, superselection rules, and quantum information,
Rev. Mod. Phys. \textbf{79}, 555 (2007).

\bibitem{MiyazakiAkibue2024}
J.~Miyazaki and S.~Akibue,
Non-locality of conjugation symmetry: Characterization and examples in quantum network sensing,
New J. Phys. \textbf{26}, 053017 (2024).

\bibitem{WuEtAl2024Distributed}
K.-D.~Wu,
T.~V.~Kondra,
C.~M.~Scandolo,
S.~Rana,
G.-Y.~Xiang,
C.-F.~Li,
G.-C.~Guo, and A.~Streltsov,
Resource theory of imaginarity in distributed scenarios,
Commun. Phys. \textbf{7}, 171 (2024).

\bibitem{GanardiEtAl2026}
R.~Ganardi,
J.~Son,
J.~Czartowski,
S.~H.~Lie, and N.~H.~Y.~Ng,
Manipulating heterogeneous quantum resources over a network,
arXiv:2602.17803.

\bibitem{BravyiKitaev2005}
S.~Bravyi and A.~Kitaev,
Universal quantum computation with ideal Clifford gates and noisy ancillas,
Phys. Rev. A \textbf{71}, 022316 (2005).

\bibitem{MarvianSpekkens2012}
I.~Marvian and R.~W.~Spekkens,
An information-theoretic account of the Wigner-Araki-Yanase theorem,
arXiv:1212.3378.

\bibitem{HickeyGour2018}
A.~Hickey and G.~Gour,
Quantifying the imaginarity of quantum mechanics,
J. Phys. A \textbf{51}, 414009 (2018).

\bibitem{WuEtAl2021}
K.-D.~Wu,
T.~V.~Kondra,
S.~Rana,
C.~M.~Scandolo,
G.-Y.~Xiang,
C.-F.~Li,
G.-C.~Guo, and A.~Streltsov,
Operational resource theory of imaginarity,
Phys. Rev. Lett. \textbf{126}, 090401 (2021).

\bibitem{WuEtAl2021PRA}
K.-D.~Wu,
T.~V.~Kondra,
S.~Rana,
C.~M.~Scandolo,
G.-Y.~Xiang,
C.-F.~Li,
G.-C.~Guo, and A.~Streltsov,
Resource theory of imaginarity: Quantification and state conversion,
Phys. Rev. A \textbf{103}, 032401 (2021).

\bibitem{Takeuchi2024}
Y.~Takeuchi,
Catalytic transformation from computationally universal to strictly universal measurement-based quantum computation,
Phys. Rev. Lett. \textbf{133}, 050601 (2024).

\bibitem{NakayamaTakeuchiAkibue2026}
Y.~Nakayama, Y.~Takeuchi, and S.~Akibue,
Uniqueness of imaginarity-assisted exact transformation from real orthogonal operations to arbitrary unitary operations,
Sci. Rep. (2026), doi:10.1038/s41598-026-70782-1.

\bibitem{Wigner1931}
E.~P.~Wigner,
Gruppentheorie und ihre Anwendung auf die Quantenmechanik der Atomspektren
(Vieweg, Braunschweig, 1931).

\bibitem{Bargmann1964}
V.~Bargmann,
Note on Wigner's theorem on symmetry operations,
J. Math. Phys. \textbf{5}, 862 (1964).

\bibitem{ChitambarGour2016}
E.~Chitambar and G.~Gour,
Critical examination of incoherent operations and a physically consistent resource theory of quantum coherence,
Phys. Rev. Lett. \textbf{117}, 030401 (2016).

\bibitem{LiuHuLloyd2017}
Z.-W.~Liu,
X.~Hu, and S.~Lloyd,
Resource destroying maps,
Phys. Rev. Lett. \textbf{118}, 060502 (2017).

\bibitem{AlbarelliEtAl2018}
F.~Albarelli,
M.~G.~Genoni,
M.~G.~A.~Paris, and A.~Ferraro,
Resource theory of quantum non-Gaussianity and Wigner negativity,
Phys. Rev. A \textbf{98}, 052350 (2018).

\bibitem{SaxenaEtAl2020}
G.~Saxena,
E.~Chitambar, and G.~Gour,
Dynamical resource theory of quantum coherence,
Phys. Rev. Research \textbf{2}, 023298 (2020).

\bibitem{TajimaTakagi2025}
H.~Tajima and R.~Takagi,
Gibbs-preserving operations requiring infinite amount of quantum coherence,
Phys. Rev. Lett. \textbf{134}, 170201 (2025).

\bibitem{NielsenChuang1997}
M.~A.~Nielsen and I.~L.~Chuang,
Programmable quantum gate arrays,
Phys. Rev. Lett. \textbf{79}, 321 (1997).

\bibitem{HilleryZimanBuzek2006}
M.~Hillery,
M.~Ziman, and V.~Bu\v{z}ek,
Approximate programmable quantum processors,
Phys. Rev. A \textbf{73}, 022345 (2006).

\bibitem{YangRennerChiribella2020}
Y.~Yang,
R.~Renner, and G.~Chiribella,
Optimal universal programming of unitary gates,
Phys. Rev. Lett. \textbf{125}, 210501 (2020).

\bibitem{HokkyoTajima2026}
A.~Hokkyo and H.~Tajima,
Quantitative Wigner-Araki-Yanase theorems for unitary and antiunitary symmetries,
arXiv:2607.09075.

\bibitem{Belenchia2013}
A.~Belenchia, G.~M.~D'Ariano, and P.~Perinotti,
Universality of computation in real quantum theory,
EPL \textbf{104}, 20006 (2013).

\bibitem{Barenco1995}
A.~Barenco, C.~H.~Bennett, R.~Cleve, D.~P.~DiVincenzo, N.~Margolus, P.~Shor,
T.~Sleator, J.~A.~Smolin, and H.~Weinfurter,
Elementary gates for quantum computation,
Phys. Rev. A \textbf{52}, 3457 (1995).

\end{thebibliography}
\end{document}